\documentclass[conference]{IEEEtran}
\IEEEoverridecommandlockouts
\usepackage{cite}
\usepackage{amsmath,amssymb,amsfonts,amsthm}
\usepackage[capitalize]{cleveref}
\usepackage{algorithm,algpseudocode}
\DeclareMathOperator*{\argmax}{arg\,max}

\usepackage{graphicx}
\usepackage{textcomp}
\usepackage{xcolor}
\usepackage{mathtools}

\newcommand{\andd}{\bigwedge}
\newcommand{\greed}{\mathcal{G}_{S}}
\newcommand{\mini}{\mathcal{M}_{S}}
\newcommand{\minithm}{\andd S}
\newcommand{\majorby}{\preceq}
\newcommand{\notmajorby}{\npreceq}
\newcommand{\allmajor}{\textsc{Majorizing-Set}}
\newcommand{\bad}{\mathcal{C}_{S}}
\newcommand{\opt}{\textsc{OPT}_{S}}

\newcommand{\mat}[1]{\mathbf{#1}}
\newcommand{\geo}{\textsc{Geom}}
\newcommand{\ind}{\textsc{Index}}
\newtheorem{definition}{Definition}

\newtheorem{corollary}{Corollary}
\newtheorem{theorem}{Theorem}
\newtheorem{lemma}{Lemma}
\newtheorem{claim}{Claim}
\newtheorem{subclaim}{Subclaim}

\def\BibTeX{{\rm B\kern-.05em{\sc i\kern-.025em b}\kern-.08em
    T\kern-.1667em\lower.7ex\hbox{E}\kern-.125emX}}
\begin{document}
\allowdisplaybreaks
\title{A Tighter Approximation Guarantee for Greedy Minimum Entropy Coupling}

\author{\IEEEauthorblockN{Spencer Compton}
\IEEEauthorblockA{\textit{MIT-IBM Watson AI Lab} \\
\textit{Massachusetts Institute of Technology}\\
Cambridge, USA \\
scompton@mit.edu}}

\maketitle

\begin{abstract}
We examine the minimum entropy coupling problem, where one must find the minimum entropy variable that has a given set of distributions $S = \{p_1, \dots, p_m \}$ as its marginals. Although this problem is NP-Hard, previous works have proposed algorithms with varying approximation guarantees. In this paper, we show that the greedy coupling algorithm of [Kocaoglu et al., AAAI'17] is always within $\log_2(e)$ ($\approx 1.44$) bits of the minimum entropy coupling. In doing so, we show that the entropy of the greedy coupling is upper-bounded by $H(\andd S) + \log_2(e)$. This improves the previously best known approximation guarantee of $2$ bits within the optimal [Li, IEEE Trans. Inf. Theory '21]. Moreover, we show our analysis is tight by proving there is no algorithm whose entropy is upper-bounded by $H(\andd S) + c$ for any constant $c<\log_2(e)$. Additionally, we examine a special class of instances where the greedy coupling algorithm is exactly optimal.
\end{abstract}

\section{Introduction}
An instance of the minimum entropy coupling problem is represented by a set $S$ of $m$ distributions, each with $n$ states (i.e., $S = \{p_1, \dots, p_m \}$). The objective is to find a variable of minimum entropy that ``couples'' $S$, meaning its marginals are equal to $S$. Equivalently, this can be described as finding a minimum entropy joint distribution over variables $p_1,\dots,p_m$.

This has a variety of applications, including areas such as causal inference \cite{kocaoglu2017entropicAAAI,kocaoglu2017entropic,compton2021entropic,javidian2021quantum} and dimension reduction \cite{vidyasagar2012metric,cicalese2016approximating}. In the context of random number generation as discussed in \cite{li2021efficient}, the minimum entropy coupling is equivalent to determining the minimum entropy variable such that one sample from this variable enables us to generate one sample from any distribution of $S$. 

While the problem is NP-Hard \cite{kovavcevic2015entropy}, previous works have designed algorithms with varying approximation guarantees. \cite{cicalese2019minimum} showed a 1-additive algorithm for $m=2$ and $\lceil \log(m) \rceil$-additive for general $m$. \cite{kocaoglu2017entropicAAAI} introduced the greedy coupling algorithm, \cite{kocaoglu2017entropic} showed this is a local optima and \cite{rossi2019greedy} showed this is a 1-additive algorithm for $m=2$. Most recently, \cite{li2021efficient} introduced a new $(2-2^{2-m})$-additive algorithm.

\emph{Our Contributions: }
Our work provides novel perspectives and analytical tools to demonstrate a tighter approximation guarantee for the greedy coupling algorithm. In \cref{section:characterization}, we show a closed-form characterization that lower-bounds each state of the greedy coupling. In \cref{section:majorizing-sets}, we study a class of instances where the greedy coupling is exactly optimal and the lower-bound characterization given in \cref{section:characterization} is tight. Finally, in \cref{section:approx} we show the greedy coupling is always within $\log_2(e)$ bits of the optimal coupling by proving it is upper-bounded by $H(\andd S) + \log_2(e)$. This improves the best-known approximation guarantee for the minimum entropy coupling problem, and we accomplish this by developing techniques involving a stronger notion of majorization and splitting distributions in an infinitely-fine manner. We show how this analysis is tight and that no algorithm can be upper-bounded by $H(\andd S) + c$ for any constant $c < \log_2(e)$. This resolves that the largest possible gap between $H(\andd S)$ and $H(\opt)$ is $\log_2(e)$.
\begin{table}[htbp]
\caption{Best-Known Additive Approximation Guarantee}
\begin{center}
\begin{tabular}{|c|c|c|c|}
\hline
\textbf{}&\multicolumn{3}{|c|}{\textbf{Algorithm (prior/now)}} \\
\cline{2-4} 
\textbf{} & \textbf{\textit{Greedy (prior)}}& \textbf{\textit{Best (prior)}}& \textbf{\textit{Greedy/Best (now)}} \\
\hline
\textbf{$m=2$}& $1$ \cite{rossi2019greedy}& $1$ \cite{cicalese2019minimum} & 1 \cite{cicalese2019minimum,rossi2019greedy} \\
\hline
\textbf{$m>2$}& $\lceil \log(m)\rceil^{\mathrm{a}}$ \cite{cicalese2019minimum,rossi2019greedy} & $2 - 2^{2-m}$ \cite{li2021efficient} & $\pmb{\log_2(e) \approx 1.44}$ \\
\hline
\multicolumn{4}{l}{$^{\mathrm{a}}$ Not explicitly shown before to our knowledge, but can combine \cite{cicalese2019minimum,rossi2019greedy}.}
\end{tabular}
\label{tab1}
\end{center}
\end{table}
\section{Background}
\emph{Notation:}
The base of $\log$ is always $2$. $H$ denotes Shannon entropy. The states of any distribution $p$ are sorted such that $p(1) \ge \dots \ge p(|p|)$. $[n]$ denotes $\{1,\dots,n \}$. $\opt$ denotes the minimum entropy coupling of a set of distributions $S$.

\emph{Greedy Minimum Entropy Coupling: } We show approximation guarantees for the greedy coupling algorithm of \cite{kocaoglu2017entropicAAAI} (formally described in \cref{alg:greedy}). At a high-level, the algorithm builds a coupling by repeatedly creating a state of the coupling output that corresponds to the currently largest state of each distribution $p_i \in S$, with weight corresponding to the smallest of these $m$ maximal states. Intuitively, this greedily adds the largest possible state to the coupling at each step. We use $\greed$ to denote the sequence of states produced by the algorithm. The algorithm runs in $O(m^2 n \log(n) )$ time.
\begin{algorithm}[ht!]
\begin{small}
    \caption{Greedy Coupling (pseudocode from \cite{kocaoglu2017entropic})}
   \label{alg:greedy}
\begin{algorithmic}[1]
    \State {\bfseries Input:} Marginal distributions of $m$ variables each with $n$ states $\{\mat{p_1},\mat{p_2},...,\mat{p_m}\}$.
    \State Initialize the tensor $\mat{P}(i_1,i_2,\hdots,i_n) = 0, \forall i_j\in [n],\forall j\in[n]$.%
    \State Initialize $r=1$.
	\While  {$r>0$} 
	\State $(\{\mat{p_i}\}_{i\in[m]}, r) = \textbf{UpdateRoutine}(\{\mat{p_i}\}_{i\in[m]}, r)$
    \EndWhile
    \State \Return $\mat{P}$.
	\State{\bfseries UpdateRoutine($\{\mat{p_1},\mat{p_2},...,\mat{p_m}\}, r$)}
    \State Find $i_j\coloneqq \argmax_{k}\{\mat{p_j}(k)\},\forall j\in[m]$.
    \State Find $u=\min\{\mat{p_k}(i_k)\}_{k\in[n]}$. \label{step:u-min}
    \State Assign $\mat{P}(i_1,i_2,\hdots,i_n) = u$.
    \State Update $\mat{p_k}(i_k)\leftarrow \mat{p_k}(i_k)-u, \forall k\in[m]$.
    \State Update $r=\sum_{k\in [n]}{\mat{p_1}(k)}$
	\State \Return $\{\mat{p_1},\mat{p_2},...,\mat{p_m}\}, r$
\end{algorithmic}
\end{small}
\end{algorithm}

\emph{Majorization: }
We use ideas from majorization theory \cite{marshall1979inequalities}. A distribution $p$ is majorized by another distribution $q$ (i.e., $p \majorby q$) if $\sum_{j=1}^{i} p(j) \le \sum_{j=1}^i q(j)$ $\forall i \in [|p|]$. It is known that if $p \majorby q$ then $H(q) \le H(p)$ \cite{marshall1979inequalities}. $\andd S$ denotes the greatest lower-bound in regards to majorization such that $\andd S \majorby p$ $\forall p \in S$. Meaning, for any $r$ where $r \majorby p$ $\forall p \in S$, it must hold that $r \majorby \andd S$. For ease of notation, we also use $\mini$ to refer to $\andd S$. It is known that $\mini(i) = \min_{p \in S} \sum_{j=1}^{i} p(j) - \sum_{j=1}^{i-1} \mini(i)$ \cite{cicalese2002supermodularity} and that $H(\andd S) \le H(\opt)$ \cite{cicalese2019minimum}.
\section{Characterization of Greedy Coupling} \label{section:characterization}
To help analyze the performance of the greedy coupling algorithm, we show this closed-form characterization that lower-bounds each element of its output:
\begin{theorem} \label{theorem:greedy-lb}
$\greed(i) \ge \max_j \frac{\sum_{k=1}^j \mini(k) - \sum_{k=1}^{i-1} \greed(k)}{j}$
\end{theorem}
\begin{proof}
We denote $p_{\ell}$ before the $t$-th step of $\greed$ as $p_{\ell}^t$. We observe that $\greed(i)$ is determined by Line 10 of \cref{alg:greedy} to be $\min_{\ell} \max_k p_{\ell}^i(k)$. We will lower-bound this quantity:

\begin{claim} \label{claim:all-p}
$\max_{1 \le k \le n} p_\ell^i(k) \ge \frac{\sum_{k=1}^j \mini(k) - \sum_{k=1}^{i-1} \greed(k)}{j}$ $\forall j,\ell$
\end{claim}
\begin{proof}

\begin{align}
    & \max_{1 \le k \le n} p_\ell^i(k) \\
    & \ge \max_{1 \le k \le j} p_\ell^i(k) \\
    & \ge \frac{\sum_{k = 1}^j p_{\ell}^i(k)}{j} \\
    & = \frac{\sum_{k = 1}^j p_\ell^1(k) - \sum_{k = 1}^j (p_\ell^1(k) - p_\ell^i(k))}{j} \\
    & \ge \frac{\sum_{k = 1}^j \mini(k) - \sum_{k = 1}^{i-1} \greed(k)}{j}
\end{align}
\end{proof}
By the definition of $\greed$ and \cref{claim:all-p}, our theorem holds.
\end{proof}
\section{Minimum Entropy Coupling of Majorizing Sets} \label{section:majorizing-sets}
Many related works show guarantees for the minimum entropy coupling problem by showing a relation to the lower-bound of $H(\andd S)$. It is natural to wonder, if we only fix $\andd S$, what is the most challenging that $S$ can be? We introduce a special-case of the minimum entropy coupling problem, where for a fixed value of $\andd S$ we consider the set $S$ to include all distributions that are consistent with $\andd S$ (i.e., all distributions that majorize $\andd S$). More formally, in this variant $S = \allmajor(p) = \{p' | p \majorby p' \}$ for some $p$. This corresponds to coupling the set of all distributions that majorize a given distribution. We show that in this setting, the greedy coupling produces the optimal solution:

\begin{theorem} \label{theorem:greedy-set-opt}
When $S = \allmajor(p)$ for some $p$, then $H(\greed) = H(\opt)$.
\end{theorem}
\begin{proof}
First, we clarify:
\begin{claim} \label{claim:mini-p}
$\mini = p$
\end{claim}
\begin{proof}
For sake of notation, suppose $p(0)=\mini(0)=0$. We will inductively show $\mini(i) = p(i)$ for all $i\in[n]$. First:
\begin{align}
    & \mini(i) \\
    & = \left(\min_{p' \in \allmajor(p)} \sum_{j=1}^{i}p'(j)\right) - \left(\sum_{j=1}^{i-1} \mini(j)\right) \\
    & \ge \left(\min_{p' \in \allmajor(p)} \sum_{j=1}^{i} p(j)\right) - \left(\sum_{j=1}^{i-1} p(j)\right) \label{step:p-lb-mini} \\
    & = p(i)
\end{align}
\eqref{step:p-lb-mini} follows as all $p' \in \allmajor(p)$ majorize $p$. Next:
\begin{align}
    & \mini(i) \\
    & = \left(\min_{p' \in \allmajor(p)} \sum_{j=1}^{i}p'(j)\right) - \left(\sum_{j=1}^{i-1}\mini(j)\right) \\
    & \le \left(\sum_{j=1}^{i} p(j)\right) - \left(\sum_{j=1}^{i-1} p(j)\right) \label{step:p-ub-mini}  = p(i)
\end{align}
\eqref{step:p-ub-mini} follows as $p \in \allmajor(p)$.
\end{proof}
We now define a distribution $\greed'$ that mirrors \cref{theorem:greedy-lb}:
\begin{definition}
$\greed'(i) = \max_j \frac{\sum_{k=1}^{j} \mini(k) - \sum_{k=1}^{i-1}\greed'(k)}{j}$
\end{definition}
Clearly $\greed'$ is a valid distribution as $\greed'(i) \le 1-\sum_{k=1}^{i-1} \greed'(k)$ and each $\greed'(i) \ge \frac{1-\sum_{k=1}^{i-1}\greed'(k)}{n}$. We show that any coupling for $S$ must be majorized by $\greed'$:
\begin{lemma} \label{lemma:major-inf}
If a distribution $\bad$ couples $S$, then $\bad \majorby \greed'$.
\end{lemma}
\begin{proof}
For sake of contradiction, suppose $\bad \notmajorby \greed'$. Then, there must exist an $i'$ where $\sum_{k=1}^{i'} \bad(k) > \sum_{k=1}^{i'} \greed'(k)$. Let $i'$ be the earliest such value. Additionally, let $j' = \argmax_j \frac{\sum_{k=1}^{j} \mini(k) - \sum_{k=1}^{i'-1}\greed'(k)}{j}$. 
We use these to define a distribution $\tilde{p} \in S$ such that $\bad$ cannot couple $\tilde{p}$:
\begin{definition}
$\tilde{p}(k)$ is $\sum_{\ell=1}^{i'} \greed'(\ell)$ for $k=1$, is $\greed'(i')$ for $1 < k \le j'$, and is $\mini(k)$ for $k > j'$
\end{definition}
\begin{claim}
$\tilde{p}$ is a valid probability distribution.
\end{claim}
\begin{proof}
All states are non-negative. Also, they sum to $1$:
\begin{align}
    & \sum_{\ell=1}^{n} \tilde{p}(\ell) \\
    & = \tilde{p}(1) + \sum_{\ell=2}^{j'} \tilde{p}(\ell) + \sum_{\ell = j'+1}^{n} \tilde{p}(\ell) \\
    & = \left(\sum_{\ell=1}^{i'} \greed'(\ell)\right) + ((j'-1) \times \greed'(i')) + \left(\sum_{\ell = j'+1}^{n} \mini(\ell)\right) \\
    & = \sum_{\ell=1}^{i'-1} \greed'(\ell) + \sum_{\ell = 1}^{j'} \mini(\ell) - \sum_{\ell=1}^{i'-1} \greed'(\ell) + \sum_{\ell = j'+1}^{n} \mini(\ell) \label{step:sub-j'} \\
    & = \sum_{\ell = 1}^{n} \mini(\ell) = 1
\end{align}
\eqref{step:sub-j'} is obtained by definition of $\greed'(i')$ and $j'$.
\end{proof}
\begin{claim}
$p \majorby \tilde{p}$
\end{claim}
\begin{proof}
We will show that $p$ is majorized by $\tilde{p}$. To begin:
\begin{subclaim} \label{subclaim:high-inds}
For $k \ge j'$, it holds that $\sum_{\ell = 1}^{k} \tilde{p}(\ell) \ge \sum_{\ell = 1}^{k} p(\ell)$
\end{subclaim}
\begin{proof}
\begin{align}
    & \sum_{\ell = 1}^{k} \tilde{p}(\ell) \\
    & = \sum_{\ell = 1}^{j'} \tilde{p}(\ell) + \sum_{\ell = j'+1}^{k} \tilde{p}(\ell) \\
    & = \left(\sum_{\ell=1}^{i'-1} \greed'(\ell) + j' \times \frac{\sum_{\ell = 1}^{j'}\mini(\ell) - \sum_{\ell = 1}^{i'-1}\greed'(\ell)}{j'}\right) \notag \\
    & + \sum_{\ell = j'+1}^{k} \mini(\ell) = \sum_{\ell = 1}^{k} \mini(\ell) \\
    & = \sum_{\ell=1}^{k} p(\ell) \label{step:use-claim2}
\end{align}
\eqref{step:use-claim2} is obtained by \cref{claim:mini-p}.
\end{proof}
Still, we must show this holds for $k < j'$. We start with:
\begin{subclaim} \label{subclaim:greed-mini}
If $j'>1$, it holds that $\greed'(i') \le \mini(j')$.
\end{subclaim}
\begin{proof}
For sake of contradiction, suppose $\greed'(i') > \mini(j')$:
\begin{align}
    & \greed'(i') \\
    & = \frac{\sum_{\ell=1}^{j'} \mini(\ell) - \sum_{\ell=1}^{i'-1} \greed'(\ell)}{j'} \\
    & = \frac{j'-1}{j'} \times \frac{\sum_{\ell=1}^{j'-1} \mini(\ell) - \sum_{\ell=1}^{i'-1} \greed'(\ell)}{j'-1} + \frac{1}{j'} \times \mini(j') \\
    & \le \frac{j'-1}{j'} \times \frac{\sum_{\ell=1}^{j'} \mini(\ell) - \sum_{\ell=1}^{i'-1} \greed'(\ell)}{j'} + \frac{1}{j'} \times \mini(j') \label{step:def-j} \\
    & = \frac{j'-1}{j'} \times \greed'(i') + \frac{1}{j'} \times \mini(j') \\
    & < \frac{j'-1}{j'} \times \greed'(i') + \frac{1}{j'} \times \greed'(i') = \greed'(i') \label{step:use-suppose}
\end{align}
This is a contradiction. \eqref{step:def-j} follows by definition of $j'$ and \eqref{step:use-suppose} by supposing $\greed'(i')>\mini(j')$.
\end{proof}
Using this, we take the next step:
\begin{subclaim} \label{subclaim:diff-change}
If $1 \le k < j'$, then $\sum_{\ell = 1}^{k} \tilde{p}(\ell) - \sum_{\ell = 1}^{k} p(\ell) \ge \sum_{\ell = 1}^{k+1} \tilde{p}(\ell) - \sum_{\ell = 1}^{k+1} p(\ell)$
\end{subclaim}
\begin{proof}
\begin{align}
    & \sum_{\ell = 1}^{k} \tilde{p}(\ell) - \sum_{\ell = 1}^{k} p(\ell) \\
    & = \sum_{\ell = 1}^{k+1} \tilde{p}(\ell) - \sum_{\ell = 1}^{k+1} p(\ell) + (p(k+1) - \tilde{p}(k+1)) \\
    & = \sum_{\ell = 1}^{k+1} \tilde{p}(\ell) - \sum_{\ell = 1}^{k+1} p(\ell) + (\mini(k+1) - \greed'(i')) \\
    & \ge \sum_{\ell = 1}^{k+1} \tilde{p}(\ell) - \sum_{\ell = 1}^{k+1} p(\ell) + (\mini(j') - \greed'(i')) \\
    & \ge \sum_{\ell = 1}^{k+1} \tilde{p}(\ell) - \sum_{\ell = 1}^{k+1} p(\ell) \label{step:use-sub2}
\end{align}
\eqref{step:use-sub2} is obtained by \cref{subclaim:greed-mini}.
\end{proof}
We now show majorization for smaller indices:
\begin{subclaim} \label{subclaim:lower-inds}
If $1 \le k < j'$, then $\sum_{\ell = 1}^{k} \tilde{p}(\ell) \ge \sum_{\ell = 1}^{k} p(\ell)$
\end{subclaim}
\begin{proof}
We can equivalently write this subclaim as how it must hold that for $1 \le k < j'$, it holds that $\sum_{\ell = 1}^{k} \tilde{p}(\ell) - \sum_{\ell = 1}^{k} p(\ell) \ge 0$. By \cref{subclaim:high-inds}, this holds for $k=j'$. By \cref{subclaim:diff-change}, the left-hand side is non-decreasing as we decrease $k$ from $j'$ to $1$. Thus, our subclaim is shown inductively.
\end{proof}
It follows from \cref{subclaim:high-inds} and \cref{subclaim:lower-inds} that $p \majorby \tilde{p}$.
\end{proof}
As we now know $\tilde{p} \in S$, we show that $\bad$ cannot couple $\tilde{p}$:
\begin{claim}
$\bad$ cannot couple $\tilde{p}$
\end{claim}
\begin{proof}
We have designed $\tilde{p}$ such that all states other than $\tilde{p}(1)$ will be too small for any of $\bad(1),\dots,\bad(i')$ to be assigned to them in a valid coupling. Additionally, we have set $\tilde{p}(1)$ to be small enough such that not all of $\bad(1),\dots,\bad(i')$ can all be assigned to $\tilde{p}(1)$ simultaneously. We prove as follows:
\begin{subclaim} \label{subclaim:bad-vs-greed}
$\bad(1) \ge \dots \ge \bad(i') > \greed'(i')$
\end{subclaim}
\begin{proof}
This holds if $\bad(i') > \greed'(i')$:
\begin{align}
    & \bad(i') \\
    & = \sum_{k = 1}^{i'} \bad(k) - \sum_{k = 1}^{i'-1} \bad(k) \\
    & \ge \sum_{k = 1}^{i'} \bad(k) - \sum_{k = 1}^{i'-1} \greed'(k) \\
    & > \sum_{k = 1}^{i'} \greed'(k) - \sum_{k = 1}^{i'-1} \greed'(k) \label{step:use-i'-def}\\
    & = \greed'(i')
\end{align}
\eqref{step:use-i'-def} is obtained by definition of $i'$.
\end{proof}
\begin{subclaim} \label{subclaim:all-assign}
For any coupling of $\tilde{p}$ with $\bad$, all of $\bad(1), \dots, \bad(i')$ must be assigned to $\tilde{p}(1)$. 
\end{subclaim}
\begin{proof}
By definition, $\tilde{p}(2),\dots,\tilde{p}(n) \ge \greed'(i')$. By \cref{subclaim:bad-vs-greed}, we then know $\bad(1) \ge \dots \ge \bad(i') > \tilde{p}(2), \dots, \tilde{p}(n)$. As such, all of $\bad(1), \dots, \bad(i')$ could only be assigned to $\tilde{p}(1)$.
\end{proof}
Further, not all of $\bad(1), \dots, \bad(i')$ can be assigned to $\tilde{p}(1)$:
\begin{subclaim} \label{subclaim:simultaneous-assign}
$\tilde{p}(1) < \sum_{k=1}^{i'} \bad(k)$
\end{subclaim}
\begin{proof}
$\tilde{p}(1) = \sum_{k=1}^{i'} \greed'(k) < \sum_{k=1}^{i'} \bad(k)$. 
\end{proof}
By \cref{subclaim:all-assign} all of $\bad(1),\dots,\bad(i')$ can only be assigned to  $\tilde{p}(1)$, yet by \cref{subclaim:simultaneous-assign} they cannot all be assigned to $\tilde{p}(1)$ simultaneously. Accordingly, $\bad$ cannot couple $\tilde{p}$.
\end{proof}
Thus, by contradiction, $\bad \majorby \greed'$ for any valid $\bad$.
\end{proof}
By \cref{lemma:major-inf}, we conclude $H(\opt) \ge H(\greed')$. Now, we show how in this setting $\greed$ is exactly $\greed'$:
\begin{lemma} \label{lemma:greed-meets-lb}
For all $i$, it holds that $\greed(i) = \greed'(i)$.
\end{lemma}
\begin{proof}
We show this inductively. Using \cref{theorem:greedy-lb} we know $\greed(i) \ge \max_j \frac{\sum_{k=1}^{j}\mini(k) - \sum_{k=1}^{i-1}\greed(k)}{j} = \max_j \frac{\sum_{k=1}^{j}\mini(k) - \sum_{k=1}^{i-1}\greed'(k)}{j} = \greed'(k)$. Using \cref{lemma:major-inf} we know $\greed(i) = \sum_{k=1}^{i} \greed(k) - \sum_{k=1}^{i-1} \greed(k) \le \sum_{k=1}^{i} \greed'(k) - \sum_{k=1}^{i-1} \greed(k) = \sum_{k=1}^{i} \greed'(k) - \sum_{k=1}^{i-1} \greed'(k) = \greed'(i)$.
\end{proof}
Thus, $H(\greed)=H(\opt)$, meaning $\greed$ is optimal.
\end{proof}
We emphasize that in \cref{lemma:greed-meets-lb} we have shown how in this setting, the characterization of \cref{theorem:greedy-lb} is actually exact.
\section{Greedy Coupling is a $\log_2(e) \approx 1.44$ Additive Approximation for Minimum Entropy Coupling} \label{section:approx}
We now show our primary result:
\begin{theorem} \label{theorem:greedy-approx}
$H(\greed) \le H \left(\minithm\right) + \log_2(e)$
\end{theorem}
\begin{proof}
We will split $\andd S$ in a particular way, and show that $\greed$ majorizes this modified distribution. Moreover, we will show that it majorizes said distribution in a very strong manner. This will enable a good approximation guarantee for $\greed$. To split $\andd S$, we introduce the geometric distribution with parameter $\gamma$ as $\geo_\gamma(x) = \gamma \times (1-\gamma)^{x-1}$. We split $\andd S$ as follows:
\begin{definition}
$\mini^{\gamma} = (\andd S) \times \geo_\gamma$
\end{definition}
We will show that $\greed$ not only majorizes $\mini^{\gamma}$ for particular $\gamma$, but also satisfies the following stronger notion:
\begin{definition}
A distribution $p$ is $\alpha$-strongly majorized by a distribution $q$ (i.e., $p \majorby_{\alpha} q$) if for all $i \in [|p|]$ there exists a $j$ such that $\sum_{k=1}^{i} p(k) \le \sum_{k=1}^{j} q(k)$ and $\alpha \times p(i) \le q(j)$.
\end{definition}
In other words, $p$ is $\alpha$-strongly majorized by $q$ if for every prefix of $p(1),\dots,p(i)$ there is a prefix of $q$ that has at least the same sum, and only contains values at least a factor of $\alpha$ greater than $p(i)$. We show that as we decrease $\gamma$ to split $\andd S$ more finely, it is increasingly strongly majorized by $\greed$:

\begin{lemma} \label{lemma:strong-major}
For any integer $z \ge 2$, $\mini^{1/z} \majorby_{z-1} \greed$
\end{lemma}
\begin{proof}
We will prove this by contradiction. Suppose that $\mini^{1/z} \notmajorby_{z-1} \greed$. This means there exists an $i,j$ such that $\sum_{k=1}^{j} \greed(k) < \sum_{k=1}^{i} \mini^{1/z}(k)$ and $\greed(j+1) < (z-1) \times \mini^{1/z}(i)$. We show that this cannot occur:
\begin{claim} \label{claim:greed-z-lb}
For integer $z\ge 2$ and any $i',j'$, if $\sum_{k=1}^{j'} \greed(k) < \sum_{k=1}^{i'} \mini^{1/z}(k)$, then $\greed(j'+1) \ge (z-1) \times \mini^{1/z}(i')$.
\end{claim}
\begin{proof}
Every element of $\mini^{1/z}$ corresponds to the product of an element of $\andd S$ and an element of $\geo_{1/z}$. We define:
\begin{definition}
$\ind_{\andd S}(k)$ is the corresponding index of $\andd S$ for $\mini^{1/z}(k)$. Likewise, $\ind_{\geo_{1/z}}(k)$ is the corresponding index of $\geo_{1/z}$ for $\mini^{1/z}(k)$.
\end{definition}
We define a set $\mathcal{T}^{i'}(k)$ for each index $k$ of $\andd S$, denoting the set of indices of $\geo_{1/z}$ in $\mini^{1/z}(1),\dots,\mini^{1/z}(i')$ corresponding to the $k$-th element of $\andd S$:
\begin{definition}
$\mathcal{T}^{i'}(k) = \{\ell | \exists i \le i' : \ind_{\andd S}(i) = k, \ind_{\geo_{1/z}}(i) = \ell \}$
\end{definition}
Also, we define the set $\mathcal{N}$ as the set of non-empty $\mathcal{T}^{i'}$:
\begin{definition}
$\mathcal{N} = \{k \in [n] | |\mathcal{T}^{i'}(k)| > 0 \}$
\end{definition}
Finally, we show our claim by:
\begin{align}
    & \greed(j'+1) \\
    & \ge \max_k \frac{\sum_{\ell = 1}^{k} \mini(\ell) - \sum_{\ell = 1}^{j'} \greed(\ell)}{k} \label{step:use-thm1}\\
    & \ge \frac{\sum_{\ell = 1}^{|\mathcal{N}|}\mini(\ell)}{|\mathcal{N}|} - \frac{\sum_{\ell = 1}^{j'} \greed(\ell)}{|\mathcal{N}|} \\
    & > \frac{\sum_{\ell = 1}^{|\mathcal{N}|}\mini(\ell)}{|\mathcal{N}|} - \frac{\sum_{\ell = 1}^{i'} \mini^{1/z}(\ell)}{|\mathcal{N}|} \label{step:use-claim-cond}\\
    & \ge \frac{1}{|\mathcal{N}|} \left(\sum_{\ell \in \mathcal{N}}\mini(\ell) - \sum_{\ell = 1}^{i'} \mini^{1/z}(\ell)\right) \\
    & = \frac{1}{|\mathcal{N}|} \sum_{\ell \in \mathcal{N}}\left(\mini(\ell) - \mini(\ell) \times \sum_{k \in \mathcal{T}^{i'}(\ell)} \geo_{1/z}(k) \right) \label{step:use-t-def} \\
    & = \frac{1}{|\mathcal{N}|} \sum_{\ell \in \mathcal{N}}\left( \sum_{k = \max(\mathcal{T}^{i'}(\ell))+1}^{\infty} \mini(\ell) \times \geo_{1/z}(k) \right) \\
    & = \frac{1}{|\mathcal{N}|} \sum_{\ell \in \mathcal{N}} \frac{(1-1/z) \times \mini(\ell) \times \geo_{1/z}(\max(\mathcal{T}^{i'}(\ell)))}{1 - (1 - 1/z)}  \\
    & \ge \frac{1}{|\mathcal{N}|} \times \sum_{\ell \in \mathcal{N}} \frac{(1-1/z) \times \mini^{1/z}(i')}{1 - (1 - 1/z)} \label{step:exists-prefix} \\
    & = (z-1) \times \mini^{1/z}(i')
\end{align}
\eqref{step:use-thm1} follows from \cref{theorem:greedy-lb}. \eqref{step:use-claim-cond} follows from the conditions of \cref{claim:greed-z-lb}. \eqref{step:use-t-def} follows by definition of $\mathcal{T}^{i'}$. \eqref{step:exists-prefix} follows from $\mini(\ell) \times \geo_{1/z}(\max(\mathcal{T}^{i'}(\ell))) \ge \mini^{1/z}(i')$ because by definition of $\mathcal{T}^{i'}$ there is an element in the prefix of $\mini^{1/z}(1),\dots,\mini^{1/z}(i')$ that corresponds to the $\ell$-th element of $\mini$ and the $\max(\mathcal{T}^{i'}(\ell))$-th element of $\geo_{1/z}$.
\end{proof}
Thus, this contradiction shows that $\mini^{1/z} \majorby_{z-1} \greed$.
\end{proof}
We could use \cref{lemma:strong-major} to immediately conclude (by setting $z=2$) that $\mini^{1/2} \majorby \greed$ and thus $H(\greed) \le H(\andd S) + 2$, giving a 2-additive approximation. However, we can do better.
\begin{lemma} \label{lemma:strong-major-ub}
If $p \majorby_{\alpha} q$, then $H(q) \le H(p) - \log(\alpha)$
\end{lemma}
\begin{proof}
For any distribution $D$, we define $\beta_D(x)$ as the set of all indices of $D$ corresponding to the minimum length prefix required to sum to at least $x$. More formally:
\begin{definition}
$\beta_D(x) = \{i \in [|D|] | \sum_{j=1}^{i-1} D(j) < x \}$
\end{definition}
With this, we show:
\begin{align}
    & H(q) \\
    & = \sum_{i=1}^{|q|} q(i) \log\left(\frac{1}{q(i)}\right) \\
    & = \sum_{i=1}^{|p|} \sum_{j \in (\beta_q(\sum_{k=1}^{i}p(k)) \backslash \beta_q(\sum_{k=1}^{i-1}p(k)))} q(j) \log\left(\frac{1}{q(j)}\right) \\
    & \le \sum_{i=1}^{|p|} \sum_{j \in (\beta_q(\sum_{k=1}^{i}p(k)) \backslash \beta_q(\sum_{k=1}^{i-1}p(k)))} q(j) \log\left(\frac{1}{\alpha \times p(i)}\right) \\
    & = \sum_{i=1}^{|p|} \log\left(\frac{1}{\alpha \times p(i)}\right) \times \sum_{j \in (\beta_q(\sum_{k=1}^{i}p(k)) \backslash \beta_q(\sum_{k=1}^{i-1}p(k)))} q(j) \\
    & \le \sum_{i=1}^{|p|} \log\left(\frac{1}{\alpha \times p(i)}\right) \times p(i) \label{step:use-summation} \\
    & = H(p) - \log(\alpha) 
\end{align}
\eqref{step:use-summation} is obtained by noticing how the sequence of the values of the inner summation must majorize $p$ by definition of $\beta_q$. As the inner summation's coefficient is non-decreasing, the equation is maximized when sequence of the values of the inner summation is exactly $p$.  
\end{proof}
\begin{corollary} \label{corr:greed-bound}
For $z \ge 2$, it holds that $H(\greed) \le H(\andd S) + H(\geo_{1/z}) - \log(z-1)$
\end{corollary}
\begin{proof}
This follows from \cref{lemma:strong-major} and \cref{lemma:strong-major-ub}.
\end{proof}
We show this upper-bound approaches $\log_2(e)$ as $z \rightarrow \infty$:
\begin{claim} \label{claim:z-limit}
$\lim_{z \rightarrow \infty} H(\geo_{1/z}) - \log(z-1) = \log_2(e)$
\end{claim}
\begin{proof}
\begin{align}
    & \lim_{z \rightarrow \infty} H(\geo_{1/z}) - \log(z-1) \\
    & = \lim_{z \rightarrow \infty} \sum_{i=0}^{\infty} \frac{(1-1/z)^i}{z} \times \log\left(\frac{z}{(1-1/z)^i}\right)- \log(z-1) \\
    & = \lim_{z \rightarrow \infty} \sum_{i=0}^{\infty} \frac{(1-1/z)^i}{z} \times i \times \log\left(\frac{1}{1-1/z}\right) + \log\left(\frac{z}{z-1}\right) \\
    & = \lim_{z \rightarrow \infty} (z-1) \times \log\left(\frac{1}{1-1/z}\right) + \log\left(\frac{z}{z-1}\right) \\
    & = \log_2(e)
\end{align}
\end{proof}
Finally, we show that $H(\greed) \le H(\andd S) + \log_2(e)$ by contradiction. Suppose there exists an $S$ where $H(\greed) = H(\andd S) + \log_2(e) + \varepsilon$ for some $\varepsilon>0$. By combining \cref{corr:greed-bound} and \cref{claim:z-limit} we can immediately conclude there is a sufficiently large $z$ where we can bound $H(\greed) < H(\andd S) + \log_2(e) + \varepsilon$. This is a contradiction, so it must hold for all $S$ that $H(\greed) \le H(\andd S) + \log_2(e)$.
\end{proof}
Moreover, this gap between $H(\greed)$ and $H(\andd S)$ is tight:
\begin{theorem} \label{theorem:greedy-tight}
There exists no algorithm $\mathcal{A}$ where it holds for all $S$ that $H(\mathcal{A}_{S}) \le H(\minithm) + c$ for any $c < \log_2(e)$.
\end{theorem}
\begin{proof}
Consider the instance $S = \allmajor(\mathcal{U}_n)$ where $\mathcal{U}_n$ is the uniform distribution over $n$ states. 
\begin{claim} \label{claim:each-greedy}
If $S=\mathcal{U}_n$, $\greed(i) = (1-1/n)^{i-1} \times 1/n$ $\forall i \ge 1$. 
\end{claim}
\begin{proof}
By \cref{lemma:greed-meets-lb}, we know $\greed(i) = \max_j \frac{\sum_{k =1}^{j}\mini(k) - \sum_{k=1}^{i-1}\greed(k)}{j} = \max_{1 \le j \le n} \frac{j/n - \sum_{k=1}^{i-1}\greed(k)}{j} = 1/n - \frac{\sum_{k=1}^{i-1} \greed(k)}{n}$. For $i=1$, $\greed(1) = 1/n - \frac{0}{n}= (1-1/n)^{0} \times 1/n$. For $i>1$ we can inductively show, $\greed(i) = 1/n - \frac{\sum_{k=1}^{i-1} \greed(k)}{n} = 1/n - \frac{n((1/n)-(1-1/n)^{i-1}/n)}{n} = 1/n - \frac{1 - (1-1/n)^{i-1}}{n} = (1-1/n)^{i-1} \times 1/n$.

\end{proof}
\begin{claim}
If $S=\mathcal{U}_n$, $\lim_{n \rightarrow \infty} H(\greed) = H(\andd S) + \log_2(e)$
\end{claim}
\begin{proof}
Using \cref{claim:each-greedy} we determine that $H(\greed) = \sum_{i=1}^{\infty} \greed(i) \times \log(\frac{1}{\greed(i)}) = \sum_{i=1}^{\infty} (1-1/n)^{i-1} \times 1/n \times \log(\frac{1}{1/n \times (1-1/n)^{i-1}}) = \log(n) + \sum_{i=1}^{\infty} (1-1/n)^{i} \times 1/n \times i \times \log(\frac{1}{1-1/n}) = \log(n) + (n-1) \times \log(\frac{n}{n-1}) = H(\andd S) + (n-1) \times \log(\frac{n}{n-1})$. Finally, $\lim_{n \rightarrow \infty} H(\greed) = H(\andd S) + \lim_{n \rightarrow \infty} (n-1) \times \log(\frac{n}{n-1}) = H(\andd S) + \log_2(e)$.
\end{proof}
By \cref{theorem:greedy-set-opt}, we know $H(\greed) = H(\opt)$. Accordingly, for any $c < \log_2(e)$ there exists an $n$ where if $S= \allmajor(\mathcal{U}_n)$ then $H(\opt) > H(\andd S) + c$.
\end{proof}
\bibliographystyle{IEEEtran}
\bibliography{isit}

\end{document}